\RequirePackage{fix-cm}
\documentclass{article}       
\usepackage{graphicx}
\usepackage{mathtools}
\usepackage{amsmath, amssymb}
\usepackage[T1]{fontenc} 
\usepackage{qcircuit}
\usepackage[misc]{ifsym}
\usepackage[margin={3cm,3cm}]{geometry}

\usepackage{changes}
\usepackage[english]{babel}
\usepackage{amsthm}
\usepackage{xfrac}
\usepackage{float}
\usepackage{wrapfig}
\usepackage{caption}
\usepackage{subcaption}
\usepackage{parskip}
\usepackage{endnotes}
\usepackage{titlesec}
\usepackage{tikz}
\usepackage{hyperref}
\usepackage{listings}
\usepackage[noload=addn]{siunitx}
\usepackage[capitalise]{cleveref}
\usepackage{booktabs}
\usepackage{import}
\usepackage{tabularx, environ}
\usepackage{qcircuit}
\usepackage{multirow} 

\newcommand{\co}[2]{#2}
\DeclarePairedDelimiter\ket{\lvert}{\rangle}
\DeclarePairedDelimiterX\braket[2]{\langle}{\rangle}{#1 \delimsize\vert #2}

\theoremstyle{plain}
\newtheorem{thm}{Theorem}[section] 

\theoremstyle{definition}
\newtheorem{defn}[thm]{Definition} 
\newtheorem{prop}[thm]{Proposition}

\makeatletter

\newcolumntype{\expand}{}
\long\@namedef{NC@rewrite@\string\expand}{\expandafter\NC@find}

\NewEnviron{problem}[2][]{%
  \def\problem@arg{#1}%
  \def\problem@framed{framed}%
  \def\problem@lined{lined}%
  \def\problem@doublelined{doublelined}%
  \ifx\problem@arg\@empty%
    \def\problem@hline{}%
  \else%
    \ifx\problem@arg\problem@doublelined%
      \def\problem@hline{\hline\hline}%
    \else%
      \def\problem@hline{\hline}%
    \fi%
  \fi%
  \ifx\problem@arg\problem@framed%
    \def\problem@tablelayout{|>{\bfseries}lX|c}%
    \def\problem@title{\multicolumn{2}{|l|}{%
        \raisebox{-\fboxsep}{\textsc{\Large #2}}%
      }}%
  \else
    \def\problem@tablelayout{>{\bfseries}lXc}%
    \def\problem@title{\multicolumn{2}{l}{%
        \raisebox{-\fboxsep}{\textsc{\Large #2}}%
      }}%
  \fi%
  \bigskip\par\noindent%
  \begin{tabularx}{\textwidth}{\expand\problem@tablelayout}%
    \problem@hline%
    \problem@title\\[2\fboxsep]%
    \BODY\\\problem@hline%
  \end{tabularx}%
  \medskip\par%
}

\makeatother

\begin{document}

\title{A polynomial size model with implicit SWAP gate counting for exact qubit reordering 
}





\author{J. Mulderij\footnote{Faculty of Electrical Engineering, Mathematics \& Computer Science, Delft University of Technology, Delft, The Netherlands} \footnote{Cyber Security \& Robustness Department, TNO, The Hague, The Netherlands} \and K.I. Aardal\footnotemark[1] \and I. Chiscop\footnotemark[2] \and F. Phillipson\footnotemark[2]
}
\date{email: \texttt{j.mulderij@tudelft.nl}}

\maketitle

\begin{abstract}
Due to the physics behind quantum computing, quantum circuit designers must adhere to the constraints posed by the limited interaction distance of qubits. Existing circuits need therefore to be modified via the insertion of SWAP gates, which alter the qubit order by interchanging the location of two qubits' quantum states. We consider the Nearest Neighbor Compliance problem on a linear array, where the number of required SWAP gates is to be minimized. We introduce an Integer Linear Programming model of the problem of which the size scales polynomially in the number of qubits and gates. Furthermore, we solve $131$ benchmark instances to optimality using the commercial solver CPLEX. The benchmark instances are substantially larger in comparison to those evaluated with exact methods before. The largest circuits contain up to $18$ qubits or over $100$ quantum gates. This formulation also seems to be suitable for developing heuristic methods since (near) optimal solutions are discovered quickly in the search process.

\end{abstract}

\section{Introduction}
\label{intro}
\co{quantum computing is hip and acquires a speedup over classical computers in some natural problems}
The rules that govern physical interactions in a quantum setting allow quantum computing to provide algorithms with a better complexity scaling than their classical counterparts for many naturally arising problems. Exploiting the properties of phenomena such as superposition and entanglement, one can search in a database \cite{grover_quantum_1997}, factor integers \cite{shor_algorithms_1994} or estimate a phase \cite{nielsen_quantum_2002} more efficiently than previously possible.

\co{Tech limitations}
The many advantages of quantum computing come at the price of physical limitations in circuit design. First, relevant coherence times (what is relevant depends on the technology) indicate that information on qubits is perturbed or even lost after some time due to a qubit's interaction with its environment \cite{divincenzo_physical_2000}. It is therefore, for a fixed number of qubits, desirable to do calculations with as few gates as possible. A second limitation is induced by nearest neighbor constraints, where 2-qubit quantum gates can only be used when the qubits are physically adjacent. The nearest neighbor constraints have been considered in proposals for a range of potential technological realizations of quantum computers such as ion traps \cite{amini_toward_2010,kumph_two-dimensional_2011,nickerson_topological_2013}, nitrogen-vacancy centers in diamonds \cite{nickerson_topological_2013,yao_quantum_2013}, quantum dots emitting linear cluster states linked by linear optics \cite{devitt_architectural_2009,herrera-marti_photonic_2010}, laser manipulated quantum dots in a cavity \cite{jones_layered_2012} and superconducting qubits \cite{divincenzo_multi-qubit_2013,ohliger_efficient_2012,linke_experimental_2017}. They are also considered in realizations of specific types of circuits and architectures, such as surface codes \cite{versluis_scalable_2017}, Shor's algorithm \cite{fowler_implementation_2004}, the Quantum Fourier Transform (QFT) \cite{takahashi_quantum_2007}, circuits for modular multiplication and exponentiation \cite{markov_constant-optimized_2012}, quantum adders on the 2D NTC architecture \cite{choi_$thetasqrtn$-depth_2012}, factoring \cite{pham_2d_2012}, fault-tolerant circuits \cite{lin_paqcs:_2015}, error correction \cite{fowler_quantum_2004}, and more recently, IBM QX architectures \cite{wille_mapping_2019,thomsen_evaluating_2019,zulehner_efficient_2019,dueck_optimization_2018}.

\co{previous work}
Up to now, the design of quantum circuits consists of manual work in elementary cases and for specific circuits. As the complexity of the algorithms increases, however, manual synthesis will no longer be feasible. When constructing a circuit from scratch, using only the set of elementary gates, even without considering nearest neighbor constraints, one is solving specific instances of the PSPACE-complete Minimum Generator Sequence problem \cite{jerrum_complexity_1985}, where the group consists of all unitary matrices and the elementary gate operations form the set of generators. Here one tries to find the shortest sequence of generators to map an input to a given output. A lot of work was done in this area  using boolean satisfiability \cite{grose_exact_2009}, template matching \cite{saeedi_synthesis_2011,maslov_quantum_2005} and methods for reversible circuits \cite{wille_considering_2014,alhagi_synthesis_2000} as all quantum gates perform unitary operations \cite{nielsen_quantum_2002}.
\co{insertion of swaps}
Other methods consider already designed circuits that do not comply with nearest neighbor constraints. In these approaches, SWAP gates, which swap the information of two adjacent qubits, are inserted into the circuit. The goal herein is to minimize the number of required SWAP gates to make the whole circuit compliant. Within this branch of research there are two approaches to the topic, global and local reordering. Global reordering determines the initial layout of the qubits such that there are as few SWAP gates as possible required in the remainder of the circuit. In order to elude the micromanagement that local reordering is concerned with, the global reordering problem is generally approximated with the NP-complete \cite{garey_computers_1979} problem of Optimal Linear Arrangement (OLA) on the interaction graph of the circuit with edge weights taking the Nearest Neighbor Cost \cite{kole_towards_2015}. Here the gate sequence is either disregarded \cite{shafaei_qubit_2014} or encoded in the weights \cite{kole_new_2018}.

\co{local reordering}
The local reordering problem allows for any change in the qubit order before each gate, resulting in a vast feasible region, even for small instances. The more general problem of SWAP minimization where qubits are placed on a coupling graph ($2$ qubits can share a gate if their corresponding nodes share an edge) is shown to be NP-Complete \cite{siraichi_qubit_2018} via a reduction from the NP-complete token swapping problem \cite{kawahara_time_2017,bonnet_complexity_2018}. The problem we consider, where the graph is a simple path, is widely believed to be NP-complete (as conjectured in \cite{hirata_efficient_2011}) but to the best of the authors' knowledge, no formal proof is given yet. 

Four research areas are distinguished in \cite{houte_mathematical_2020}, each corresponding to either local or global reordering and to either a single quantum computer or a network thereof. In \cite{houte_mathematical_2020}, the focus lies on networks of quantum computers, relating to the field of distributed quantum computing. This work proposes a new model for the local reordering problem on a single quantum computer. Many heuristics have been developed in this area of research, including receding horizon \cite{kole_heuristic_2016,shafaei_optimization_2013,wille_look-ahead_2016,hirata_efficient_2011}, greedy \cite{hirata_efficient_2011,alfailakawi_harmony-search_2016}, harmony search \cite{alfailakawi_harmony-search_2016} and OLA on parts of the circuit \cite{pedram_layout_2016}. Only a few works have dared to approach the problem with exact methods, all of which embody an explicit factorial scaling in the number of variables or processed nodes either through the use of the adjacent transposition graph \cite{matsuo_changing_2012}, exhaustive searches \cite{ding_exact_2019,hirata_efficient_2011} or explicit cost enumeration for each permutation \cite{wille_exact_2014}. The exact approaches have delivered small benchmark instances to compare the heuristics' results to. The size of these benchmark instances typically does not exceed circuits of about $5$ qubits and $16$ gates due to the vast scaling of the number of variables in the optimization model.

\co{our contribution}
In this work we will provide an exact Integer Linear Programming (ILP) formulation of the Nearest Neighborhor Compliance (NNC) problem that does not entail a factorial scaling in the number of qubits, by implicitly counting the number of required SWAP gates at each reordering step. The power of the commercial optimization solver CPLEX \cite{manual1987ibm} is used to optimally solve the problem for $123$ instances from the \textit{RevLib} library \cite{wille_revlib:_2008} and $8$ QFT circuits. The considered benchmark instances include the largest circuits to be exactly solved up to this point. They include the QFT for $10$ qubits and even a circuit with $18$ qubits. The evaluation of the bigger benchmark instances finally allows for heuristics to be compared to exact solutions on larger circuits.

\co{structure of letter}
The remainder of this paper is structured as follows. In Sec.~\ref{sec:background} we introduce basic concepts of quantum computing. In Sec.~\ref{sec:problemdefinition} the problem of NNC is formulated. Next, in Sec.~\ref{sec:model}, the proposed mathematical model is introduced. The results are presented and discussed in Sec.~\ref{sec:results}. Finally, conclusions are drawn in Sec.~\ref{sec:conclusion}.

\section{Background}
\label{sec:background}
\co{Here the preliminaries will be discussed.}
In this section we will first introduce some basic concepts of quantum computing. A more detailed explanation can be found in \cite{nielsen_quantum_2002}. Then, a description of decomposing multi-qubit gates is given.
\subsection{Building blocks of QC}
\label{sec:backbulidingblocks}
\co{Qubit states are represented by complex vectors, superposition basic notation.}
The quantum version of the classical basic unit of computation, the bit, is the quantum bit (qubit). The qubit has the special property that it does not have to take value 0 or 1, but it can be in a superposition of the computational basis states $\ket{0} \equiv [1,0]^T$ and $\ket{1} \equiv [0,1]^T$. The state of a qubit $\ket{\phi}$ is denoted by a vector in $\mathbb{C}^2$ where in general we write

\begin{equation}
\ket{\phi} = \alpha \ket{0} + \beta \ket{1},
\end{equation}
where $\alpha,\beta \in \mathbb{C}$. When information about the state's value is extracted by the means of measurement, the state collapses to a single value. If, for example, the measurement is done in the standard basis, one would obtain $\ket{0}$ with probability $|\alpha|^2$ and $\ket{1}$ with probability $|\beta|^2$. Necessarily, $|\alpha|^2 + |\beta|^2 = 1$. In $n$-qubit systems, the combined state is the tensor product of individual states, which is an element of $\mathbb{C}^{2^n}$. Calculations are done by executing quantum circuits, which consist of a set of qubits and a list of quantum gates. The initial qubit states are the input of the calculation. The gates operate, in order, on specified qubits. Afterwards, a measurement is performed on one or more of the qubits to determine the probabilistic outcome of the calculation. Quantum gates are inherently reversible and are denoted by linear operators in the form of invertible matrices. Their action on the combined qubit state is simply the matrix vector product.

Below we will introduce some of the most common quantum gates, starting with the controlled NOT (CNOT) gate, see Fig.~\ref{fig:CNOT}.

\begin{figure*}
\centering
\begin{minipage}[c]{0.5\textwidth}
	\centering
	{\Qcircuit @C=1em @R=1.5em {
			\lstick{q_1} & \ctrl{1} & \qw & \rstick{q_1}\\
			\lstick{q_2} & \targ & \qw & \rstick{q_2 \oplus q_1}
	}}
\end{minipage}
\caption{A CNOT gate. Qubit $q_1$ is the control qubit and $q_2$ is the target qubit.} 
\label{fig:CNOT}       
\end{figure*}
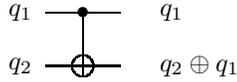

The controlled CNOT gate is one of the most commonly used gates. It is also used to construct the SWAP gate by placing three CNOT gates consecutively such as in Fig.~\ref{fig:SWAP}.

\begin{figure*}[h]
\noindent
\begin{minipage}[c]{0.5\textwidth}
	\centering
	{\Qcircuit @C=1em @R=1.5em {
			\lstick{q_1} & \ctrl{1} & \targ & \ctrl{1} & \qw & \rstick{q_2}\\
			\lstick{q_2} & \targ & \ctrl{-1} & \targ & \qw & \rstick{q_1}
	}}
\end{minipage}%
\centering
\hspace{1cm}
$\iff$
\hspace{1cm}
\begin{minipage}[c]{0.5\textwidth}
	\centering
	{\Qcircuit @C=1em @R=1.5em {
			\lstick{q_1} & \qswap & \qw & \rstick{q_2}\\
			\lstick{q_2} & \qswap\qwx & \qw & \rstick{q_1}
	}}
\end{minipage}
\caption{A decomposed and composite SWAP gate. The operations are equivalent, the gates interchange the states of two qubits.} 
\label{fig:SWAP}       
\end{figure*}
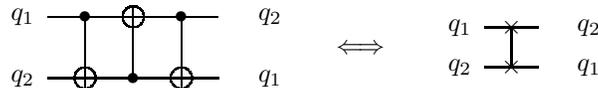
The SWAP gate will be the main tool used to overcome the physical constraints that limit quantum circuit design. We will, for the remainder of this text, not make a distinction between interchanging two qubits and interchanging the quantum states of two qubits.

\subsection{Decomposing multi-qubit gates}
\label{sec:backdecomposing}
\co{many circuits have multi controlled not gates for example, give some libraries and their gates. this is how we go back to 2-qubit gates only}
Many circuits make use of composite gates that resemble entire circuits themselves. They are often performed on more than two qubits at once, take for example the quantum Fourier Transform (QFT), which can act on any number of qubits. In order to describe what it means for a gate to act on adjacent qubits, it only makes sense to consider 2-qubit gates. To achieve this without losing the meaning of the circuit, we have to do a modification in the following two cases:

\begin{enumerate}
\item Gates that only act on a single qubit are ignored for the rest of this research. These gates are of no interest in this context.
\item Gates that act on more than two qubits are decomposed into 2-qubit gates. The fact that this is always possible can be found in \cite{nielsen_quantum_2002}.
\end{enumerate}

The second point can be implemented in a great variety of ways and doing this ``optimally" is outside the scope of this work. We therefore make two straightforward design choices: 1) We only consider circuits using multiple-control Toffoli gates, Peres gates and multiple-control Fredkin gates up to a certain size; 2) We always decompose a given circuit in the same way. There is clearly room for improvement here, but the search space we consider is large enough as it is. We ignore all single-qubit gates during the modification to a nearest neighbor compliant circuit. The normal Toffoli gate's decomposition, with two control qubits can be found in \cite{barenco_elementary_1995} in the section ``Three-Bit Networks''. In the same work, the decomposition of a 3-control Toffoli gate is shown in the section ``$n$-Bit Networks''. The decomposition of the 4-control Toffoli is the direct extension of the previous decompositions. The Peres gate is decomposed as in the circuit ``peres\_8.real'' from RevLib \cite{wille_revlib:_2008}. The Fredkin gate is decomposed as in the circuit ``fredkin\_5.real", also from RevLib. The two-qubit controlled Fredkin gate is decomposed into a controlled-NOT gate, a Toffoli gate and another controlled-NOT gate as shown by \cite{alhagi_synthesis_2000} in Fig 2.4c. Larger composed gates do not make an appearance in the circuits that are considered in this work.

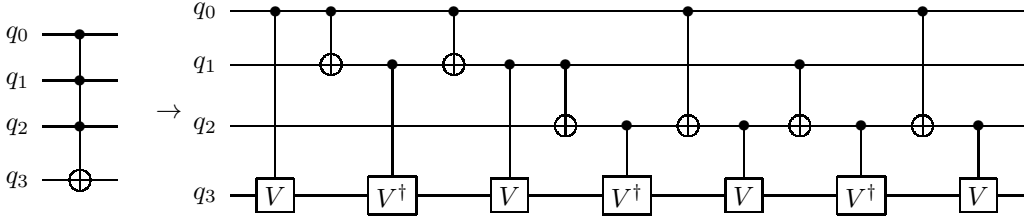
\begin{figure*}[h]
\noindent
\begin{minipage}[c]{0.5\textwidth}
	\centering
	{\Qcircuit @C=1em @R=1.5em {
			\lstick{q_0} & \ctrl{1} & \qw\\
			\lstick{q_1} & \ctrl{1} & \qw\\
			\lstick{q_2} & \ctrl{1} & \qw\\
			\lstick{q_3} & \targ    & \qw
	}}
\end{minipage}%
\centering
\hspace{0.4cm}
$\rightarrow$
\hspace{0.4cm}
\begin{minipage}[c]{0.2\textwidth}
	\centering
	{\Qcircuit @C=1em @R=1.5em {
			\lstick{q_0} & \ctrl{3} & \ctrl{1} & \qw           & \ctrl{1} & \qw      & \qw      & \qw           & \ctrl{2} & \qw      & \qw      & \qw           & \ctrl{2} & \qw  &    \qw \\
			\lstick{q_1} & \qw		& \targ    & \ctrl{2}      & \targ    & \ctrl{2} & \ctrl{1} & \qw           & \qw      & \qw      & \ctrl{1} & \qw           & \qw      & \qw  &    \qw \\
			\lstick{q_2} & \qw 		& \qw      & \qw           & \qw      & \qw      & \targ    & \ctrl{1}      & \targ    & \ctrl{1} & \targ    & \ctrl{1}      & \targ    & \ctrl{1}  &    \qw \\
			\lstick{q_3} & \gate{V} & \qw      & \gate{V^\dag} & \qw      & \gate{V} & \qw      & \gate{V^\dag} & \qw      & \gate{V} & \qw      & \gate{V^\dag} & \qw      & \gate{V} & \qw
	}}
\end{minipage}
\caption{The decomposition of a Toffoli gate with $3$ control qubits and one target qubit into only 2-qubit gates. Here we have $V^4 = X$, where $X$ is the usual Pauli-X gate.} 
\label{fig:CCCNOT}       
\end{figure*}

\section{Problem Definition}
\label{sec:problemdefinition}
\co{This section is about formalizing the problem}
In this section, some basic definitions will be introduced in order to formalize the NNC problem.

\co{Definitions lets go.}
For the NNC problem, the actual operation corresponding to a gate that is being used has no influence on the problem. Only the qubits on which the gate acts matter. Some definitions are introduced below. Denote the set $Q$ of $n$ qubits as the set of integers $Q = \{1,\hdots,n\}$. Since all the qubits have one physical location in a one dimensional array, the locations are numbered as $L = (1,\hdots,n)$ and are in a fixed order. To keep track of the location of each qubit before every gate, the notion of a qubit order will be introduced. 

\begin{defn}
Let $\mathcal{S}_n$ be the permutation group and $[n]$ the vector $(1,\hdots,n)$. Then a \textit{qubit order} is a permutation denoted by the vector $\tau ([n])$ with $\tau \in \mathcal{S}_n$, which maps the qubits to locations. We call $\tau^t$ the qubit order before gate $t$.
\end{defn}

Now that the qubit orders are defined, one needs a way of altering such an order. This is done via the previously mentioned SWAP gates.

\begin{defn}
A \textit{SWAP gate} is an adjacent transposition $\tau \in \mathcal{S}_n$ that permutes a qubit order, $\tau \circ (q_1,\hdots,q_i,q_{i+1},\hdots,q_{n}) = (q_1,\hdots,q_{i+1},q_i,\hdots,q_{n})$, by interchanging the positions of two adjacent qubits.
\end{defn}

The number of SWAP gates that one minimally requires to ``move'' from one qubit order to another is inherently equal to the Kendall tau distance between the corresponding permutations.

\begin{defn}
Given two permutations $\tau_1,\tau_2 \in \mathcal{S}_n$ for some fixed $n$, the \textit{Kendall tau} distance between $\tau_1$ and $\tau_2$ is defined as
\begin{equation}
I(\tau_1,\tau_2) \equiv | \left\{ (i,j) \mid 1\leq i,j\leq n, \tau_1(i) < \tau_1(j), \tau_2(i) > \tau_2(j) \right\} |.
\end{equation}
\end{defn}

This metric counts the number of inversions between two orderings $\tau$ of items. It states that the number of adjacent transpositions required to sort the array is equal to the number of inversions in the array.
The nearest neighbor interaction constraints can only be formulated once the concept of quantum gates has been properly introduced in this setting.

\begin{defn}
Let $q_i,q_j \in Q$ be two qubits such that $i\neq j$. Let $g_{ij}$ be an unordered pair $g = \{q_i,q_j\}$. Then we say that $g_{ij}$ is a \textit{quantum gate}, or simply a \textit{gate}, that acts on qubits $q_i$ and $q_j$. When the specific qubits do not matter in the context, the subscripts may be omitted. When multiple gates are present and their order is important, this will be reflected with a superscript as $g^t$.
\end{defn}
Please note that this definition only allows for quantum gates that act on pairs of qubits. If a gate (in the more general sense) acts on more qubits, we assume it to be decomposed, whilst if it only works on one qubit, the gate can be ignored.

To describe an entire quantum circuit, multiple gates are needed and their order is important. To this end, a gate sequence is introduced.

\begin{defn}
Let $g^1,\hdots,g^{m}$ be $m$ gates. Let $G$ be the finite sequence of gates $G = (g^1,\hdots,g^{m})$, then we say $G$ is a \textit{gate sequence} of size $m$.
\end{defn}

We also assume the gate sequence to be given and fixed. Allowing changes in the gate order when some commutative rules are satisfied, as was done in \cite{matsuo_changing_2012,itoko_quantum_2019,hattori_quantum_2018}, is beyond the scope of this work.

Now we can introduce the concept of a quantum circuit more formally.

\begin{defn}
Let $Q$ be the set of qubits and $G$ be a gate sequence. Let $QC$ be a tuple of the set of qubits and the gate sequence $QC=(Q,G)$. Then we say that $QC$ is a \textit{quantum circuit}.
\end{defn}

At the core of the problem are the nearest neighbor (NN) constraints. Formalizing these requires a number of the above definitions. These constraints are what make the problem difficult.

\begin{defn}
Given are a gate $g_{ij}^t$ and a qubit order $\tau^t$ before that gate. We say that the gate \textit{complies with the NN constraints} if $|\tau(i)-\tau(j)|=1$, i.e., if the qubits on which the gate acts are adjacent in the qubit order. If, given a qubit order for each gate, all the gates in a quantum circuit's gate sequence comply with the NN constraints, we say that the quantum circuit complies with the NN constraints.
\end{defn}

Now that all these concepts have been formalized, we can continue with defining the problem of NNC.

\begin{problem}[doublelined]{Nearest Neighbor Compliance Problem}
\textbf{Input}: & A quantum circuit $QC = (Q,G)$ with $|Q|=n$ qubits and $|G|=m$ gates and an integer $k\in \mathbb{Z}_{\geq 0}$.\\
\textbf{Question}: & Do there exist qubit orders $\tau^t, t\in[m]$, one before each gate of $QC$, such that the sum of the Kendall tau distances between consecutive qubit orders satisfies $\sum_{t=1}^{m-1} I(\tau^t,\tau^{t+1}) \leq k$ and such that the quantum circuit complies with the NN constraints?
\end{problem}
In the minimization version of the problem, which we model in the next section, we seek to find the smallest integer $k$ such that Problem 1 is still answered affirmatively. Considering the problem in this way, we do not require the qubits to end up in the same qubit order as they started out in. We also do not allow for changes in the gate order and do not optimize over different ways of decomposing multi-qubit quantum gates. The objective function in the minimization problem simply counts the number of required SWAP gates.

Note that calculating the Kendall tau distance between two permutations can be naively done in $\mathcal{O}(n^2)$ time, following the steps of the bubble sort algorithm \cite{knuth_art_1974}. A faster computation of the distance, in $\mathcal{O}(n\sqrt{\log n})$ time, can be found in \cite{chan_counting_2010}.

We will however not be concerned with explicitly listing the Kendall tau distances for all $n!$ permutations. In order to avoid the listing, the metric should be implicitly calculated in the model. The objective function, variables and constraints that allow us to do so, will be introduced in the next section.

\section{Mathematical Model}
\label{sec:model}
\co{here we will present the mathematical model in terms of vars, constraints, objective}
In this section the proposed ILP formulation of the NNC minimization problem will be discussed in detail. First, the variables and constraints are presented and explained. Finally, the complete model is given, along with a linearization of the constraints.

Given a quantum circuit $QC = (Q,G)$, we introduce integer variables $x_i^t\in L=\{1,\hdots,n\}$ for the location of each qubit $i\in Q$ before each gate $g^t \in G$. Since the goal is to avoid the explicit $n!$ scaling in the number of variables and constraints, we make use of the Kendall tau metric to count the number of required SWAP gates when going from one qubit order $\tau^t$ to the next $\tau^{t+1}$. To accomplish this, keeping track of the pairwise order of the qubits is essential. We introduce binary variables to do precisely this,

\begin{equation}
y_{ij}^t  =
\begin{cases}
    1       & \quad \text{if location } x_i^t \text{ is before location } x_j^t \text{ in qubit order } \tau^t\\
    0  & \quad \text{else.}
\end{cases}
\end{equation}

The $y$-variables are only defined for $i<j$, so that every pair of qubits is only compared once. Keeping track of changes in the $y$-variables when moving from one qubit order to the next allows us to count the number of SWAP gates needed. The $x$- and $y$-variables are related through the following big-$M$ type constraints,

\begin{align}
x_i^t - x_j^t &\leq My_{ij}^t - 1 &\quad \forall i,j\in Q, i<j, t\in [m] \label{eq:bigM1}\\
x_j^t - x_i^t &\leq M(1-y_{ij}^t) - 1 &\quad \forall i,j\in Q, i<j, t\in [m] \label{eq:bigM2}
\end{align}
where $M$ is a big enough constant, $M=(n+1)$ being sufficient in this case. Note that these constraints also enforce two important features:
\begin{enumerate}
\item No two qubits can be at the same location at the same time.
\item The definition of the $y$ variables is enforced by the constraints.
\end{enumerate}
For fixed $i,j$ and $t$, one of the two constraints is always trivially satisfied due to the large value of $M$. The $-1$ term in the right-hand side even ensures that the location indices differ by at least one from each other. This allows us, later on, to relax the $x$ variables to be continuous without losing the property that feasible solutions have integer $x$ variables.

To make sure that the result also complies with the NN constraints, the following constraints need to be added:

\begin{align}
-1 \leq x_i^t - x_j^t &\leq 1 &\quad\forall g_{ij}^t \in G \label{eq:gate1}
\end{align}
For each gate that acts on qubits $q_i$ and $q_j$, the qubit order that is assumed just before the gate, it is required to have the qubits in adjacent locations.

The objective is to minimize the total number of absolute changes in the $y$ variables,

\begin{equation}\label{eq:objold}
\min \displaystyle\sum_{\substack{i,j\in Q\\i<j}}\displaystyle\sum_{t\in [m-1]} |y_{ij}^t - y_{ij}^{t+1}|.
\end{equation}
Notice that the objective function exactly computes the Kendall tau distance between every two consecutive qubit orders. Currently, the objective function is not linear. Extra binary variables $k_{ij}^t$ are introduced to linearize the objective function. These substitute $|y_{ij}^t - y_{ij}^{t+1}|$ in the objective function and are constrained in the following manner

\begin{align}
-k_{ij}^t \leq y_{ij}^t - y_{ij}^{t+1} &\leq k_{ij}^t &\quad \forall i,j\in Q, i<j, t\in [m-1] \label{eq:count1}
\end{align}
Now the $k$-variables can be substituted into Expression \eqref{eq:objold}, which, together with the constraints, result in the ILP model:

\begin{equation}
\begin{aligned}
&\min & & \displaystyle\sum_{\substack{i,j\in Q\\i<j}}\displaystyle\sum_{t\in [m-1]} k_{ij}^t \\
&\text{subject to } & & \eqref{eq:bigM1},\eqref{eq:bigM2},\eqref{eq:gate1},\eqref{eq:count1}\\
& & & x_i^t\in \{1,\hdots,n\} &\quad\forall i\in Q, t\in[m] \\
& & &y_{ij}^t\in \{0,1\} &\quad\forall i,j\in Q,i<j, t\in[m] \\
& & &k_{ij}^t\in \{0,1\} &\quad\forall i,j\in Q,i<j, t\in[m-1]
\end{aligned}
\end{equation}
Simply counting the number of variables in this formulation gives
\begin{equation}
\# \text{ variables} = n^2m-\frac{n^2-n}{2},
\end{equation}
and the number of constraints in the ILP is equal to
\begin{equation}
\# \text{ constraints} = 2(n^2-n)m - n^2 + n + 2m,
\end{equation}
which is polynomial in the number of qubits and gates.
In order to improve running times in practice, it helps to relax variables to take continuous values. We state the following about this relaxation:
\begin{prop}
Allowing the $x$- and $k$-variables to take continuous values does not change the optimal value.
\end{prop}
\begin{proof}
The $x$-variables must take values that are pairwise separated from each other by at least $1$ due to constraints $\eqref{eq:bigM1},\eqref{eq:bigM2}$. There are $n$ variables that all have to take a value in a connected interval of length $n$, all spaced at least $1$ from each other. This can only be done if the $x$'s are all integer and all integer values are taken.
The $k$-variables are constrained by $\eqref{eq:count1}$. Since the $y$-variables are binary, their difference is also binary (or $-1$, in which case $k=0$ is allowed). Since we are minimizing over the $k$-variables, their value will always assume the smallest possible allowed value by the constraints, which is integer.
\end{proof}
Even though relaxing these variables does not impact the objective value of optimal solutions, it reduces the number of integer-restricted variables which improves the running time in practice. This stems from the underlying fact that it is an NP-complete problem to find an optimal solution to a general ILP, while doing so for a linear program (LP) is polynomially solvable with interior point methods.

\section{Experimental Results}
\label{sec:results}
\co{Here we provide tables filled with evaluations of benchmark instances}
In this section, the results of evaluating the proposed ILP model are presented. The time of finding the optimal solution in the proposed ILP model is compared to the time required by the previous best exact approaches. The attained objective value of a multitude of heuristic approaches is also compared with the solution that our model provides.
\subsection{Experimental setup}
\label{sec:resultssetup}
\co{What did we use to get the results?}
The mathematical model as described in the previous section has been implemented in Python and solved with the commercial solver CPLEX 12.7 through the Python API. All but the quantum Fourier transform instances, which were constructed following the circuit of \cite{nielsen_quantum_2002}, were obtained form the RevLib \cite{wille_revlib:_2008} website. The evaluations were conducted using up to $16$ threads of $2.4$ GHz each, working with 16 GB of RAM. All instances were solved to optimality.

The benchmark instances are subdivided over three tables, according to the number of qubits addressed. In the first column of each table, the name of the circuit is provided, and in the second column, $n$ denotes the number of qubits in the circuit. In the third column, $|G|$ denotes the number of 2-qubit gates present in the circuit after gate decomposition and the removal of single-qubit gates. The optimal value of the local reordering problem, i.e., the minimum number of needed SWAP gates to make the circuit nearest neighbor compliant, is provided in the fourth column. The column ``Time'' denotes the run time in seconds. The column entitled ``Time E'' denotes the running time of other exact methods, also in seconds. Exact running times with subscript $a$ are from \cite{wille_exact_2014}, subscript $b$ from \cite{matsuo_changing_2012}. Heuristic solution's objective values are presented in the last column, denoted by ``$\#$ SWAPS H''. Here the subscript $c$ indicates the results are from \cite{kole_new_2018}, subscript $d$ from \cite{shafaei_optimization_2013}, subscript $e$ from \cite{alfailakawi_line_2013}, subscript $f$ from \cite{kole_heuristic_2016} and subscript $g$ from \cite{wille_look-ahead_2016}. An asterisk as superscript indicates that for the other exact solution methods, either the objective value differs, or the number of gates differs or they both differ. For the heuristic results, the asterisk indicates that the number of gates differs or the objective value of the heuristic is lower than that of the proposed exact method. These anomalies are believed to find their roots in differing gate decomposition methods, resulting in slightly different instances.
\subsection{Results}
\co{which results did we get?}
The running time required to solve the instance is heavily dependent on three factors:
\begin{enumerate}
\item The number of qubits in the quantum circuit,
\item The number of gates in the quantum circuit,
\item The minimal number of required SWAP gates.
\end{enumerate}

The number of qubits and gates is expected to heavily influence the running time. The number of qubits is the term that influences the run time the most. This is due to the fact that the number of feasible solutions scales factorially in the number of qubits.
Surprisingly, the run time also scales quite badly with the number of required SWAP gates. During the Branch \& Bound tree search, the upper bound determined by CPLEX, which is the best feasible solution found up to that point, converges to the optimal value (or close to it) rather quickly. The best known lower bound, however, takes a long time to improve. When the number of required SWAP gates increases, the time needed to improve the lower bound all the way to the optimal value increases as well. This phenomenon is analyzed for two of the benchmark instances that require a lot of SWAP gates:
\begin{enumerate}
\item \textbf{mod8-10\_177} The search method found a feasible solution with an objective value within $10\%$ of the optimal value in $2.4\cdot 10^6$ iterations, found an optimal solution in $4.0\cdot 10^7$ iterations, and proved optimality by a matching lower bound after $1.7 \cdot 10^8$ iterations.
\item \textbf{decod24-enable\_126} The search method found a feasible solution with an objective value within $10\%$ of the optimal value in $5.3\cdot 10^6$ iterations, found an optimal solution in $1.0\cdot 10^7$ iterations, and proved optimality by a matching lower bound after $5.8 \cdot 10^7$ iterations.
\end{enumerate}
If a $10\%$ optimality gap would suffice, only less than $2\%$ of the total number of iterations would be needed in the first case, and $10\%$ in the second case. This observation indicates that running an incomplete Branch and Bound algorithm might be an interesting and easy-to-implement heuristic algorithm.

The $131$ evaluated benchmark instances are listed in the tables below. The improvement in computation time with respect to previous exact methods is significant. The results have been compared to other exact and heuristic methods. There is no standard set of benchmark instances, so not every benchmark instance has been evaluated with every method. Sometimes methods slightly differ in the problem they are solving by allowing alterations of the quantum circuit as a prepossessing step for example. The latter point may result in slightly differing optimal solutions, this is indicated with an asterisk in the tables. In the table, the running time of solving our integer linear programming model is compared to the running time of the other exact solution methods. 

Other exact solution methods can only solve the smaller circuits, making it impossible to compare their performance as the circuit size increases, apart from the binary statement that our method can indeed solve the instance. The results show exact solutions that are obtained for much larger circuits than previously held possible. The largest instance with respect to the number of qubits has as much as $18$ qubits. Furthermore, for the first time, NNC has been solved to optimality for circuits with more than $100$ quantum gates.

We also compare our algorithm to existing heuristic solution methods. These are much faster than our exact solution method, but do not guarantee optimal solutions to the problem. The comparison is most interesting in Tab.~\ref{tab:3}, where the considered quantum circuits require a higher number of SWAP gates to comply with the NN constraints. Here we see that the heuristic methods have an optimality gap of $42.1\%$ averaged over the comparable benchmark instances in Tab.~\ref{tab:3}.

\begin{table}
\caption{Benchmark instances with three or four qubits.}
\begin{tabular}{lllllll}
\hline\noalign{\smallskip}
Benchmark & $n$ & $|G|$ & $\#$ SWAPS & Time & Time E & $\#$ SWAPS H\\
\noalign{\smallskip}\hline\noalign{\smallskip}
QFT\_QFT3 & $3$ & $3$ & $1$ & $0.02$ & - & - \\
peres\_10 & $3$ & $4$ & $1$ & $0.14$ & $0.1_{a}$ & -\\
peres\_8 & $3$ & $4$ & $1$ & $0.06$ & $0.1_{a}$ & -\\
toffoli\_2 & $3$ & $5$ & $1$ & $0.12$ & $0.2_{a}$ & -\\
toffoli\_1 & $3$ & $5$ & $1$ & $0.1$ & $0.1_{a}$ & -\\
peres\_9 & $3$ & $6$ & $1$ & $0.02$ & $2463_{a}$ & -\\
fredkin\_7 & $3$ & $7$ & $1$ & $0.16$ & - & -\\
ex-1\_166 & $3$ & $7$ & $2$ & $0.08$ & $0.1_{a}$ & -\\
fredkin\_5 & $3$ & $7$ & $1$ & $0.15$ & $0.1_{a},0.1_{b}^*$ & -\\
ham3\_103 & $3$ & $8$ & $2$ & $0.04$ & - & -\\
miller\_12 & $3$ & $8$ & $2$ & $0.14$ & $745.6_{a},0.1_{b}$ & -\\
ham3\_102 & $3$ & $9$ & $1$ & $0.05$ & $0.1_{a}^*$ & -\\
3\_17\_15 & $3$ & $9$ & $2$ & $0.04$ & $630.2_{a},0.1_{b}^*$ & -\\
3\_17\_13 & $3$ & $13$ & $3$ & $0.12$ & $0.1_{a}^*$ & $4_{c}^*,4_{d},3_{e},6_{g}$\\
3\_17\_14 & $3$ & $13$ & $3$ & $0.15$ & $0.1_{a}^*$ & -\\
fredkin\_6 & $3$ & $15$ & $3$ & $0.06$ & $4.6_{a}$ & -\\
miller\_11 & $3$ & $17$ & $4$ & $0.15$ & $0.1_{a}^*$ & -\\
QFT\_QFT4 & $4$ & $6$ & $3$ & $0.17$ & - & -\\
toffoli\_double\_3 & $4$ & $7$ & $1$ & $0.11$ & $0.9_{a},0.1_{b}^*$ & -\\
rd32-v1\_69 & $4$ & $8$ & $2$ & $0.16$ & $0.1_{a}$ & -\\
decod24-v1\_42 & $4$ & $8$ & $2$ & $0.12$ & $7.7_{a},0.1_{b}^*$ & -\\
rd32-v0\_67 & $4$ & $8$ & $2$ & $0.07$ & $1.6_{a}$ & $2_{c},2_{d}$\\
decod24-v2\_44 & $4$ & $8$ & $3$ & $0.07$ & $0.1_{b}^*$ & -\\
decod24-v0\_40 & $4$ & $8$ & $3$ & $0.06$ & $0.1_{b}^*$ & -\\
decod24-v3\_46 & $4$ & $9$ & $3$ & $0.09$ & $0.1_{a},0.1_{b}^*$ & $3_{c},3_{d}$\\
toffoli\_double\_4 & $4$ & $10$ & $2$ & $0.07$ & $200_{a}^2$ & -\\
rd32-v1\_68 & $4$ & $12$ & $3$ & $0.24$ & $0.4_{a}^*$ & -\\
rd32-v0\_66 & $4$ & $12$ & $0$ & $0.09$ & $0.4_{a}^*$ & -\\
decod24-v0\_39 & $4$ & $15$ & $5$ & $0.53$ & $0.5_{a}$ & -\\
decod24-v2\_43 & $4$ & $16$ & $5$ & $0.23$ & $0.1_{a}^*$ & -\\
decod24-v0\_38 & $4$ & $17$ & $4$ & $0.57$ & $19.2_{a}$ & -\\
decod24-v1\_41 & $4$ & $21$ & $7$ & $0.5$ & - & -\\
hwb4\_52 & $4$ & $23$ & $8$ & $0.97$ & - & $9_{c},10_{d},9_{e},9_{f}$\\
aj-e11\_168 & $4$ & $29$ & $12$ & $5.36$ & - & -\\
4\_49\_17 & $4$ & $30$ & $12$ & $6.1$ & - & $12_{c}^*,12_{d},16_{e}$\\
decod24-v3\_45 & $4$ & $32$ & $13$ & $6.25$ & - & -\\
mod10\_176 & $4$ & $42$ & $15$ & $7.94$ & - & -\\
aj-e11\_165 & $4$ & $44$ & $18$ & $9.36$ & - & $36_{d},33_{g}^*$\\
mod10\_171 & $4$ & $57$ & $24$ & $27.18$ & - & -\\
4\_49\_16 & $4$ & $59$ & $22$ & $24.23$ & - & -\\
mini-alu\_167 & $4$ & $62$ & $27$ & $23.7$ & - & -\\
hwb4\_50 & $4$ & $63$ & $23$ & $17.61$ & - & -\\
hwb4\_49 & $4$ & $65$ & $23$ & $21.64$ & - & -\\
hwb4\_51 & $4$ & $75$ & $28$ & $75.09$ & - & -\\
\noalign{\smallskip}\hline
\end{tabular}
\end{table}

\newpage

\begin{table}
\caption{Benchmark instances with five qubits}
\begin{tabular}{lllllll}
\hline\noalign{\smallskip}
Benchmark & $n$ & $|G|$ & $\#$ SWAPS & Time & Time E & $\#$ SWAPS H\\
\noalign{\smallskip}\hline\noalign{\smallskip}
4mod5-v1\_25 & $5$ & $7$ & $1$ & $0.26$ & $11705.3_{a}$ & -\\
4gt11\_84 & $5$ & $7$ & $1$ & $0.06$ & $16.6_{a}$ & $1_{c},1_{d},1_{e}$\\
4gt11-v1\_85 & $5$ & $7$ & $1$ & $0.09$ & - & -\\
4mod5-v0\_20 & $5$ & $8$ & $2$ & $0.08$ & $45.5_{a}$ & -\\
4mod5-v1\_22 & $5$ & $9$ & $1$ & $0.08$ & $548.8_{a}^*$ & -\\
QFT\_QFT5 & $5$ & $10$ & $6$ & $0.41$ & $1.6_{a}$ & $7_{c},6_{d}$\\
mod5d1\_63 & $5$ & $11$ & $2$ & $0.12$ & - & -\\
4mod5-v0\_19 & $5$ & $12$ & $3$ & $0.84$ & $55.3_{a}^*$ & -\\
4gt11\_83 & $5$ & $12$ & $3$ & $0.15$ & $9_{a}^*$ & -\\
4mod5-v1\_24 & $5$ & $12$ & $3$ & $0.28$ & - & -\\
mod5mils\_65 & $5$ & $12$ & $4$ & $0.26$ & - & -\\
mod5mils\_71 & $5$ & $12$ & $2$ & $0.15$ & - & -\\
alu-v2\_33 & $5$ & $13$ & $4$ & $0.45$ & - & -\\
alu-v1\_29 & $5$ & $13$ & $4$ & $0.61$ & - & -\\
alu-v0\_27 & $5$ & $13$ & $4$ & $0.48$ & - & -\\
mod5d2\_70 & $5$ & $14$ & $5$ & $0.43$ & - & -\\
alu-v3\_35 & $5$ & $14$ & $5$ & $0.38$ & - & -\\
alu-v4\_37 & $5$ & $14$ & $5$ & $0.37$ & - & -\\
alu-v1\_28 & $5$ & $14$ & $4$ & $0.26$ & - & -\\
4gt13-v1\_93 & $5$ & $15$ & $5$ & $0.69$ & $489.3_{a}^*$ & $7_{c}^*,6_{d},4_{e}^*$\\
4gt13\_92 & $5$ & $15$ & $6$ & $0.53$ & - & -\\
4gt11\_82 & $5$ & $16$ & $6$ & $0.89$ & - & -\\
4mod5-v0\_21 & $5$ & $17$ & $8$ & $2.84$ & - & -\\
rd32\_272 & $5$ & $18$ & $7$ & $0.94$ & - & -\\
alu-v3\_34 & $5$ & $18$ & $4$ & $0.4$ & - & -\\
mod5d2\_64 & $5$ & $19$ & $6$ & $1.81$ & - & -\\
alu-v0\_26 & $5$ & $21$ & $8$ & $3.56$ & - & -\\
4gt5\_75 & $5$ & $21$ & $6$ & $1.1$ & - & $9_{c}^*,12_{d}$\\
4mod5-v0\_18 & $5$ & $23$ & $8$ & $3.35$ & - & -\\
4mod5-v1\_23 & $5$ & $24$ & $9$ & $5.06$ & - & $9_{c},9_{d},15_{e}$\\
one-two-three-v2\_100 & $5$ & $24$ & $7$ & $5.37$ & - & -\\
one-two-three-v3\_101 & $5$ & $24$ & $7$ & $2.96$ & - & -\\
rd32\_271 & $5$ & $26$ & $11$ & $7.37$ & - & -\\
4gt5\_77 & $5$ & $28$ & $10$ & $6.2$ & - & -\\
4gt5\_76 & $5$ & $29$ & $10$ & $5.45$ & - & -\\
alu-v4\_36 & $5$ & $30$ & $9$ & $6.34$ & - & $15_{c}^*,18_{d},17_{e}$\\
4gt13\_91 & $5$ & $30$ & $8$ & $4.46$ & - & -\\
4gt13\_90 & $5$ & $34$ & $12$ & $6.77$ & - & -\\
4gt10-v1\_81 & $5$ & $34$ & $13$ & $12.38$ & - & $18_{c}^*,20_{d},16_{e},24_{g}^*$\\
one-two-three-v1\_99 & $5$ & $36$ & $15$ & $17.27$ & - & -\\
4gt4-v0\_80 & $5$ & $36$ & $19$ & $43.45$ & - & $34_{d},33_{f}$\\
4mod7-v0\_94 & $5$ & $38$ & $12$ & $12.83$ & - & -\\
alu-v2\_32 & $5$ & $38$ & $16$ & $22.05$ & - & -\\
4mod7-v0\_95 & $5$ & $38$ & $14$ & $14.59$ & - & $19_{c}^*,21_{d},22_{e}$\\
4mod7-v1\_96 & $5$ & $38$ & $14$ & $13.49$ & - & -\\
\noalign{\smallskip}\hline
\end{tabular}
\end{table}

\begin{table}
\caption{Benchmark instances with five qubits and more than 40 gates}
\begin{tabular}{lllllll}
\hline\noalign{\smallskip}
one-two-three-v0\_98 & $5$ & $40$ & $15$ & $15.67$ & - & -\\
4gt12-v0\_88 & $5$ & $41$ & $20$ & $34.01$ & - & -\\
4gt12-v1\_89 & $5$ & $44$ & $22$ & $52.36$ & - & $35_{d},26_{e},32_{f}$\\
sf\_275 & $5$ & $46$ & $18$ & $21.42$ & - & -\\
4gt4-v0\_79 & $5$ & $49$ & $22$ & $80.16$ & - & -\\
4gt4-v0\_78 & $5$ & $53$ & $26$ & $167.03$ & - & -\\
4gt4-v0\_72 & $5$ & $53$ & $24$ & $49.7$ & - & -\\
4gt12-v0\_87 & $5$ & $54$ & $22$ & $45.88$ & - & -\\
4gt4-v1\_74 & $5$ & $57$ & $29$ & $84.87$ & - & -\\
4gt12-v0\_86 & $5$ & $58$ & $26$ & $108.35$ & - & -\\
mod8-10\_178 & $5$ & $68$ & $37$ & $389.47$ & - & -\\
one-two-three-v0\_97 & $5$ & $71$ & $32$ & $76.8$ & - & -\\
4gt4-v0\_73 & $5$ & $89$ & $40$ & $699.65$ & - & -\\
mod8-10\_177 & $5$ & $93$ & $48$ & $3650.26$ & - & $72_{d}$\\
alu-v2\_31 & $5$ & $100$ & $49$ & $2906.35$ & - & -\\
hwb5\_55 & $5$ & $101$ & $48$ & $2264.0$ & - & $59_{c},63_{d},60_{e},66_{g}$\\
rd32\_273 & $5$ & $104$ & $50$ & $4631.7$ & - & -\\
alu-v2\_30 & $5$ & $112$ & $55$ & $13558.87$\\
\noalign{\smallskip}\hline
\end{tabular}
\end{table}

\begin{table}
\caption{Benchmark instances with six or more qubits.} \label{tab:3}
\begin{tabular}{lllllll}
\hline\noalign{\smallskip}
Benchmark & $n$ & $|G|$ & $\#$ SWAPS & Time & Time E & $\#$ SWAPS H\\
\noalign{\smallskip}\hline\noalign{\smallskip}
graycode6\_47 & $6$ & $5$ & $0$ & $0.02$ & - & -\\
graycode6\_48 & $6$ & $5$ & $0$ & $0.02$ & - & -\\
QFT\_QFT6 & $6$ & $15$ & $11$ & $7.43$ & - & $11_{c},12_{d}$\\
decod24-enable\_124 & $6$ & $21$ & $5$ & $1.86$ & - & -\\
decod24-enable\_125 & $6$ & $21$ & $5$ & $1.83$ & - & -\\
decod24-bdd\_294 & $6$ & $24$ & $7$ & $9.37$ & - & -\\
mod5adder\_129 & $6$ & $71$ & $34$ & $534.38$ & - & -\\
mod5adder\_128 & $6$ & $77$ & $36$ & $1103.51$ & - & $45_{c}^*,51_{d},46_{g}^*$\\
decod24-enable\_126 & $6$ & $86$ & $37$ & $1954.28$ & - & -\\
xor5\_254 & $7$ & $5$ & $3$ & $0.61$ & - & -\\
ex1\_226 & $7$ & $5$ & $3$ & $0.25$ & - & -\\
QFT\_QFT7 & $7$ & $21$ & $16$ & $28.26$ & - & $28_{c},26_{d},18_{g}$\\
4mod5-bdd\_287 & $7$ & $23$ & $7$ & $4.3$ & - & -\\
alu-bdd\_288 & $7$ & $28$ & $8$ & $20.65$ & - & -\\
ham7\_106 & $7$ & $49$ & $28$ & $495.43$ & - & -\\
ham7\_105 & $7$ & $65$ & $34$ & $1613.33$ & - & -\\
ham7\_104 & $7$ & $83$ & $42$ & $3238.82$ & - & $56_{c}^*$\\
QFT\_QFT8 & $8$ & $28$ & $23$ & $334.6$ & - & $32_{c},33_{d},31_{g}$\\
rd53\_139 & $8$ & $36$ & $11$ & $76.29$ & - & -\\
rd53\_138 & $8$ & $44$ & $11$ & $100.86$ & - & -\\
rd53\_137 & $8$ & $66$ & $35$ & $6271.11$ & - & -\\
QFT\_QFT9 & $9$ & $36$ & $30$ & $1482.53$ & - & $52_{c},54_{d},49_{g}$\\
QFT\_QFT10 & $10$ & $45$ & $39$ & $39594.99$ & - & $64_{g}$\\
mini\_alu\_305 & $10$ & $57$ & $23$ & $1711.75$ & - & -\\
sys6-v0\_144 & $10$ & $62$ & $19$ & $887.71$ & - & -\\
rd73\_141 & $10$ & $64$ & $21$ & $845.05$ & - & -\\
parity\_247 & $18$ & $16$ & $14$ & $5762.29$ & - & -\\
\noalign{\smallskip}\hline
\end{tabular}
\end{table}
\vfill
\newpage

\hfill
\hfill
\hfill
\hfill
\section{Conclusion}
\label{sec:conclusion}
\co{What did we do? What can we say about the results? What were the limitations? What are interesting next research topics?}
In this paper we consider the local reordering scheme for nearest neighbor architectures of quantum circuits. We propose a new mathematical model that counts the number of required SWAP gates implicitly, by using specific properties of the constraints. The implicit counting improves upon previous exact approaches in which costs were explicitly determined for each permutation, leading to a factorial scaling of the model size, and therefore, a high running time. The presented innovations result in a great improvement in the model size, such that the resulting ILP only contains $\mathcal{O}(n^2m)$ variables and constraints.

The benchmark instances with available exact solutions known in the literature were no larger than circuits with five qubits and no more than twenty gates, due to the excessive running times. The proposed method can handle quantum circuits with five qubits and $112$ gates or up to eighteen qubits and sixteen gates. In total $131$ benchmark instances are evaluated, most of which have not been solved to optimality.

Because the implicit counting is based on counting inversions in permutations, the formulation is not easily translated to the popular higher dimensional cases where qubits are placed on a $2$D or $3$D grid. To the authors' best knowledge there is no known polynomial time algorithm that, in 2- or 3-dimensional grids, solves the subproblem of calculating the minimum number of required SWAP gates when transforming one qubit order into another. Such a method could have great impact on exact solution methods in the higher-dimensional setting.

Practical experience with the Branch \& Bound tree search indicates that finding a (near) optimal feasible solution does not consume the most computation time. This means that solving the ILP heuristically, with a restriction in running time or iteration count for example, could make for a good heuristic solution method.

\bibliographystyle{spmpsci}      

\begin{thebibliography}{10}
\providecommand{\url}[1]{{#1}}
\providecommand{\urlprefix}{URL }
\expandafter\ifx\csname urlstyle\endcsname\relax
  \providecommand{\doi}[1]{DOI~\discretionary{}{}{}#1}\else
  \providecommand{\doi}{DOI~\discretionary{}{}{}\begingroup
  \urlstyle{rm}\Url}\fi

\bibitem{alfailakawi_harmony-search_2016}
AlFailakawi, M.G., Ahmad, I., Hamdan, S.: Harmony-search algorithm for 2d
  nearest neighbor quantum circuits realization.
\newblock Expert Syst. with Appl. \textbf{61}, 16--27 (2016)

\bibitem{alfailakawi_line_2013}
AlFailakawi, M.G., AlTerkawi, L., Ahmad, I., Hamdan, S.: Line ordering of
  reversible circuits for linear nearest neighbor realization.
\newblock Quantum Inf. Process. \textbf{12}(10), 3319--3339 (2013)

\bibitem{alhagi_synthesis_2000}
Alhagi, N.: Synthesis of {Reversible} {Functions} {Using} {Various} {Gate}
  {Libraries} and {Design} {Specifications}.
\newblock Tech. rep., Portland State University (2000)

\bibitem{amini_toward_2010}
Amini, J.M., Uys, H., Wesenberg, J.H., Seidelin, S., Britton, J., Bollinger,
  J.J., Leibfried, D., Ospelkaus, C., VanDevender, A.P., Wineland, D.J.: Toward
  scalable ion traps for quantum information processing.
\newblock New J. Phys. \textbf{12}(3), 033031 (2010)

\bibitem{barenco_elementary_1995}
Barenco, A., Bennett, C.H., Cleve, R., DiVincenzo, D.P., Margolus, N., Shor,
  P., Sleator, T., Smolin, J., Weinfurter, H.: Elementary gates for quantum
  computation.
\newblock Phys. Rev. A \textbf{52}(5), 3457--3467 (1995)

\bibitem{bonnet_complexity_2018}
Bonnet, E., Miltzow, T., Rzążewski, P.: Complexity of {Token} {Swapping} and
  {Its} {Variants}.
\newblock Algorithmica \textbf{80}(9), 2656--2682 (2018)

\bibitem{chan_counting_2010}
Chan, T.M., Pătraşcu, M.: Counting {Inversions}, {Offline} {Orthogonal}
  {Range} {Counting}, and {Related} {Problems}.
\newblock In: Proceedings of the {Twenty}-{First} {Annual} {ACM}-{SIAM}
  {Symposium} on {Discrete} {Algorithms}, pp. 161--173. Society for Industrial
  and Applied Mathematics (2010)

\bibitem{choi_$thetasqrtn$-depth_2012}
Choi, B.S., Van~Meter, R.: An
  \${\textbackslash}{Theta}({\textbackslash}sqrt\{n\})\$-depth {Quantum}
  {Adder} on a 2d {NTC} {Quantum} {Computer} {Architecture}.
\newblock J. Emerg. Technol. Comput. Syst. \textbf{8}(3), 1--22 (2012)

\bibitem{manual1987ibm}
Cplex, IBM ILOG: V12. 1: User’s Manual for CPLEX.
\newblock International Business Machines Corporation, \textbf{46}(53) (2009)

\bibitem{devitt_architectural_2009}
Devitt, S.J., Fowler, A.G., Stephens, A.M., Greentree, A.D., Hollenberg,
  L.C.L., Munro, W.J., Nemoto, K.: Architectural design for a topological
  cluster state quantum computer.
\newblock New J. Phys. \textbf{11}(8), 083032 (2009)

\bibitem{ding_exact_2019}
Ding, J., Yamashita, S.: Exact {Synthesis} of {Nearest} {Neighbor} {Compliant}
  {Quantum} {Circuits} in 2d architecture and its {Application} to
  {Large}-scale {Circuits}.
\newblock IEEE Trans. on Comput.-Aided Des. of Integr. Circuits and Syst. pp.
  1--1 (2019)

\bibitem{divincenzo_physical_2000}
DiVincenzo, D.P., IBM: The {Physical} {Implementation} of {Quantum}
  {Computation}.
\newblock Fortschr. der Phys. \textbf{48}(9-11), 771--783 (2000)

\bibitem{divincenzo_multi-qubit_2013}
DiVincenzo, D.P., Solgun, F.: Multi-qubit parity measurement in circuit quantum
  electrodynamics.
\newblock New J. Phys. \textbf{15}(7), 075001 (2013)

\bibitem{dueck_optimization_2018}
Dueck, G.W., Pathak, A., Rahman, M.M., Shukla, A., Banerjee, A.: Optimization
  of {Circuits} for {IBM}'s five-qubit {Quantum} {Computers}.
\newblock pp. 680--684 (2018)

\bibitem{fowler_implementation_2004}
Fowler, A.G., Devitt, S.J., Hollenberg, L.C.L.: Implementation of {Shor}'s
  {Algorithm} on a {Linear} {Nearest} {Neighbour} {Qubit} {Array}.
\newblock arXiv:quant-ph/0402196  (2004).
\newblock ArXiv: quant-ph/0402196

\bibitem{fowler_quantum_2004}
Fowler, A.G., Hill, C.D., Hollenberg, L.C.L.: Quantum {Error} {Correction} on
  {Linear} {Nearest} {Neighbor} {Qubit} {Arrays}.
\newblock Phys. Rev. A \textbf{69}(4), 042314 (2004)

\bibitem{garey_computers_1979}
Garey, M.R., Johnson, D.S.: Computers and {Intractability}; {A} {Guide} to the
  {Theory} of {NP}-{Completeness}.
\newblock W. H. Freeman \& Co., New York, NY, USA (1979)

\bibitem{grose_exact_2009}
Große, D., Wille, R., Dueck, G.W., Drechsler, R.: Exact {Multiple}-{Control}
  {Toffoli} {Network} {Synthesis} {With} {SAT} {Techniques}.
\newblock IEEE Trans. on Comput.-Aided Des. of Integr. Circuits and Syst.
  \textbf{28}(5), 703--715 (2009)

\bibitem{grover_quantum_1997}
Grover, L.K.: Quantum {Mechanics} {Helps} in {Searching} for a {Needle} in a
  {Haystack}.
\newblock Phys. Rev. Lett. \textbf{79}(2), 325--328 (1997)

\bibitem{hattori_quantum_2018}
Hattori, W., Yamashita, S.: Quantum {Circuit} {Optimization} by {Changing} the
  {Gate} {Order} for 2d {Nearest} {Neighbor} {Architectures}.
\newblock In: J.~Kari, I.~Ulidowski (eds.) Reversible {Computation}, Lecture
  {Notes} in {Computer} {Science}, pp. 228--243. Springer International
  Publishing (2018)

\bibitem{herrera-marti_photonic_2010}
Herrera-Martí, D.A., Fowler, A.G., Jennings, D., Rudolph, T.: A {Photonic}
  {Implementation} for the {Topological} {Cluster} {State} {Quantum}
  {Computer}.
\newblock Phys. Rev. A \textbf{82}(3), 032332 (2010)

\bibitem{hirata_efficient_2011}
Hirata, Y., Nakanishi, M., Yamashita, S., Nakashima, Y.: {An} {Efficient}
  {Conversion} {of} {Quantum} {Circuits} {to} {a} {Linear} {Nearest} {Neighbor}
  {Architecture}.
\newblock Quantum Inf. and Comput. \textbf{11}(1\&2), 25 (2011)

\bibitem{houte_mathematical_2020} 
van Houte, R., Mulderij, J., Attema, T., Chiscop, I., Phillipson, F.: Mathematical formulation of quantum circuit design problems in networks of quantum computers. 
\newblock Quantum Information Processing, 19(5), 1-22 (2020).

\bibitem{itoko_quantum_2019}
Itoko, T., Raymond, R., Imamichi, T., Matsuo, A., Cross, A.W.: Quantum circuit
  compilers using gate commutation rules.
\newblock In: Proceedings of the 24th {Asia} and {South} {Pacific} {Design}
  {Automation} {Conference} on - {ASPDAC} '19, pp. 191--196. ACM Press, Tokyo,
  Japan (2019)

\bibitem{jerrum_complexity_1985}
Jerrum, M.R.: {The} {Complexity} {of} {Finding} {Minimum}-{Length} {Generator}
  {Sequences}.
\newblock Theor. Comput. Sci. \textbf{36}, 25 (1985)

\bibitem{jones_layered_2012}
Jones, N.C., Van~Meter, R., Fowler, A.G., McMahon, P.L., Kim, J., Ladd, T.D.,
  Yamamoto, Y.: Layered {Architecture} for {Quantum} {Computing}.
\newblock Phys. Rev. X \textbf{2}(3), 031007 (2012)

\bibitem{kawahara_time_2017}
Kawahara, J., Saitoh, T., Yoshinaka, R.: The {Time} {Complexity} of the {Token}
  {Swapping} {Problem} and {Its} {Parallel} {Variants}.
\newblock In: S.H. Poon, M.S. Rahman, H.C. Yen (eds.) {WALCOM}: {Algorithms}
  and {Computation}, Lecture {Notes} in {Computer} {Science}, pp. 448--459.
  Springer International Publishing (2017)

\bibitem{knuth_art_1974}
Knuth, D.E.: The {Art} of {Computer} {Programming}, {Volume} 3: {Sorting} and
  {Searching}, {Second} {Edition}, vol.~3, 2nd edn. (1974)

\bibitem{kole_heuristic_2016}
Kole, A., Datta, K., Sengupta, I.: A {Heuristic} for {Linear} {Nearest}
  {Neighbor} {Realization} of {Quantum} {Circuits} by {SWAP} {Gate} {Insertion}
  {Using}\${N}\$-{Gate} {Lookahead}.
\newblock IEEE J. on Emerg. and Sel. Top. in Circuits and Syst. \textbf{6}(1),
  62--72 (2016)

\bibitem{kole_new_2018}
Kole, A., Datta, K., Sengupta, I.: A {New} {Heuristic} for \${N}\$
  -{Dimensional} {Nearest} {Neighbor} {Realization} of a {Quantum} {Circuit}.
\newblock IEEE Trans. on Comput.-Aided Des. of Integr. Circuits and Syst.
  \textbf{37}(1), 182--192 (2018)

\bibitem{kole_towards_2015}
Kole, A., Datta, K., Sengupta, I., Wille, R.: Towards a {Cost} {Metric} for
  {Nearest} {Neighbor} {Constraints} in {Reversible} {Circuits}.
\newblock Rev. Comput. \textbf{9138}, 273--278 (2015)

\bibitem{kumph_two-dimensional_2011}
Kumph, M., Brownnutt, M., Blatt, R.: Two-dimensional arrays of radio-frequency
  ion traps with addressable interactions.
\newblock New J. Phys. \textbf{13}(7), 073043 (2011)

\bibitem{lin_paqcs:_2015}
Lin, C., Sur-Kolay, S., Jha, N.K.: {PAQCS}: {Physical} {Design}-{Aware}
  {Fault}-{Tolerant} {Quantum} {Circuit} {Synthesis}.
\newblock IEEE Trans. on Very Large Scale Int. Syst. \textbf{23}(7), 1221--1234
  (2015)

\bibitem{linke_experimental_2017}
Linke, N.M., Maslov, D., Roetteler, M., Debnath, S., Figgatt, C., Landsman,
  K.A., Wright, K., Monroe, C.: Experimental comparison of two quantum
  computing architectures.
\newblock Proc Natl Acad Sci USA \textbf{114}(13), 3305--3310 (2017)

\bibitem{markov_constant-optimized_2012}
Markov, I.L., Saeedi, M.: Constant-{Optimized} {Quantum} {Circuits} for
  {Modular} {Multiplication} and {Exponentiation}.
\newblock arXiv:1202.6614 [quant-ph]  (2012).
\newblock ArXiv: 1202.6614

\bibitem{maslov_quantum_2005}
Maslov, D., Young, C., Miller, D., Dueck, G.: Quantum {Circuit}
  {Simplification} {Using} {Templates}.
\newblock In: Design, {Automation} and {Test} in {Europe}, pp. 1208--1213.
  IEEE, Munich, Germany (2005)

\bibitem{matsuo_changing_2012}
Matsuo, A., Yamashita, S.: Changing the {Gate} {Order} for {Optimal} {LNN}
  {Conversion}.
\newblock In: A.~De~Vos, R.~Wille (eds.) Reversible {Computation}, Lecture
  {Notes} in {Computer} {Science}, pp. 89--101. Springer Berlin Heidelberg
  (2012)

\bibitem{nickerson_topological_2013}
Nickerson, N.H., Li, Y., Benjamin, S.C.: Topological quantum computing with a
  very noisy network and local error rates approaching one percent.
\newblock Nat. Commun. \textbf{4}(1) (2013)

\bibitem{nielsen_quantum_2002}
Nielsen, M.A., Chuang, I., Grover, L.K.: Quantum {Computation} and {Quantum}
  {Information}.
\newblock Am. J. of Phys. \textbf{70}(5), 558--559 (2002)

\bibitem{ohliger_efficient_2012}
Ohliger, M., Eisert, J.: Efficient measurement-based quantum computing with
  continuous-variable systems.
\newblock Phys. Rev. A \textbf{85}(6), 062318 (2012)

\bibitem{pedram_layout_2016}
Pedram, M., Shafaei, A.: Layout {Optimization} for {Quantum} {Circuits} with
  {Linear} {Nearest} {Neighbor} {Architectures}.
\newblock IEEE Circuits and Syst. Mag. \textbf{16}(2), 62--74 (2016)

\bibitem{pham_2d_2012}
Pham, P., Svore, K.M.: A 2d {Nearest}-{Neighbor} {Quantum} {Architecture} for
  {Factoring} in {Polylogarithmic} {Depth}.
\newblock arXiv:1207.6655 [quant-ph]  (2012).
\newblock ArXiv: 1207.6655

\bibitem{saeedi_synthesis_2011}
Saeedi, M., Wille, R., Drechsler, R.: Synthesis of quantum circuits for linear
  nearest neighbor architectures.
\newblock Quantum Inf. Process. \textbf{10}(3), 355--377 (2011)

\bibitem{shafaei_optimization_2013}
Shafaei, A., Saeedi, M., Pedram, M.: Optimization of quantum circuits for
  interaction distance in linear nearest neighbor architectures.
\newblock In: 2013 50th {ACM}/{EDAC}/{IEEE} {Design} {Automation} {Conference}
  ({DAC}), pp. 1--6 (2013)

\bibitem{shafaei_qubit_2014}
Shafaei, A., Saeedi, M., Pedram, M.: Qubit placement to minimize communication
  overhead in 2d quantum architectures.
\newblock In: 2014 19th {Asia} and {South} {Pacific} {Design} {Automation}
  {Conference} ({ASP}-{DAC}), pp. 495--500 (2014)

\bibitem{shor_algorithms_1994}
Shor, P.: Algorithms for quantum computation: discrete logarithms and
  factoring.
\newblock In: Proceedings 35th {Annual} {Symposium} on {Foundations} of
  {Computer} {Science}, pp. 124--134. IEEE Comput. Soc. Press, Santa Fe, NM,
  USA (1994)

\bibitem{siraichi_qubit_2018}
Siraichi, M.Y., Santos, V.F.d., Collange, S., Pereira, F.M.Q.: Qubit
  {Allocation}.
\newblock In: Proceedings of the 2018 {International} {Symposium} on {Code}
  {Generation} and {Optimization}, {CGO} 2018, pp. 113--125. ACM, New York, NY,
  USA (2018).
\newblock Event-place: Vienna, Austria

\bibitem{takahashi_quantum_2007}
Takahashi, Y., Kunihiro, N., Ohta, K.: The {Quantum} {Fourier} {Transform} on a
  {Linear} {Nearest} {Neighbor} {Architecture}.
\newblock Quantum Info. Comput. \textbf{7}(4), 383--391 (2007)

\bibitem{versluis_scalable_2017}
Versluis, R., Poletto, S., Khammassi, N., Haider, N., Michalak, D.J., Bruno,
  A., Bertels, K., DiCarlo, L.: Scalable quantum circuit and control for a
  superconducting surface code.
\newblock Phys. Rev. Applied \textbf{8}(3), 034021 (2017)

\bibitem{wille_mapping_2019}
Wille, R., Burgholzer, L., Zulehner, A.: Mapping {Quantum} {Circuits} to {IBM}
  {QX} {Architectures} {Using} the {Minimal} {Number} of {SWAP} and {H}
  {Operations}.
\newblock In: Proceedings of the 56th {Annual} {Design} {Automation}
  {Conference} 2019 on - {DAC} '19, pp. 1--6. ACM Press, Las Vegas, NV, USA
  (2019)

\bibitem{wille_revlib:_2008}
Wille, R., Große, D., Teuber, L., Dueck, G.W., Drechsler, R.: {RevLib}: {An}
  {Online} {Resource} for {Reversible} {Functions} and {Reversible} {Circuits}.
\newblock In: 38th {International} {Symposium} on {Multiple} {Valued} {Logic}
  (ismvl 2008), pp. 220--225 (2008)

\bibitem{wille_look-ahead_2016}
Wille, R., Keszocze, O., Walter, M., Rohrs, P., Chattopadhyay, A., Drechsler,
  R.: Look-ahead schemes for nearest neighbor optimization of 1d and 2d quantum
  circuits.
\newblock In: 2016 21st {Asia} and {South} {Pacific} {Design} {Automation}
  {Conference} ({ASP}-{DAC}), pp. 292--297. IEEE, Macao, Macao (2016)

\bibitem{wille_considering_2014}
Wille, R., Lye, A., Drechsler, R.: Considering nearest neighbor constraints of
  quantum circuits at the reversible circuit level.
\newblock Quantum Inf. Process. \textbf{13}(2), 185--199 (2014)

\bibitem{wille_exact_2014}
Wille, R., Lye, A., Drechsler, R.: Exact {Reordering} of {Circuit} {Lines} for
  {Nearest} {Neighbor} {Quantum} {Architectures}.
\newblock IEEE Trans. on Comput.-Aided Des. of Integr. Circuits and Syst.
  \textbf{33}(12), 1818--1831 (2014)

\bibitem{yao_quantum_2013}
Yao, N.Y., Gong, Z.X., Laumann, C.R., Bennett, S.D., Duan, L.M., Lukin, M.D.,
  Jiang, L., Gorshkov, A.V.: Quantum {Logic} between {Remote} {Quantum}
  {Registers}.
\newblock Phys. Rev. A \textbf{87}(2), 022306 (2013)

\bibitem{thomsen_evaluating_2019}
Zulehner, A., Bauer, H., Wille, R.: Evaluating the {Flexibility} of {A}* for
  {Mapping} {Quantum} {Circuits}.
\newblock In: M.K. Thomsen, M.~Soeken (eds.) Reversible {Computation}, vol.
  11497, pp. 171--190. Springer International Publishing, Cham (2019)

\bibitem{zulehner_efficient_2019}
Zulehner, A., Paler, A., Wille, R.: An {Efficient} {Methodology} for {Mapping}
  {Quantum} {Circuits} to the {IBM} {QX} {Architectures}.
\newblock IEEE Trans. on Comput.-Aided Des. of Integr. Circuits and Syst.
  \textbf{38}(7), 1226--1236 (2019)

\end{thebibliography}

%
%

\end{document}